\documentclass[sn-mathphys]{sn-jnl}

\usepackage{hyperref}
\hypersetup{ colorlinks=true, linkcolor=blue, citecolor=blue,
  filecolor=blue, urlcolor=blue}
 \usepackage{bbm}
 \usepackage{caption}
 \usepackage{subcaption}
 \usepackage{hhline}
 \usepackage{adjustbox}

\usepackage{graphicx}%
\usepackage{multirow}%
\usepackage{amsmath,amssymb,amsfonts}%
\usepackage{amsthm}%
\usepackage{mathrsfs}%
\usepackage[title]{appendix}%
\usepackage{xcolor}%
\usepackage{textcomp}%
\usepackage{manyfoot}%
\usepackage{booktabs}%
\usepackage{algorithm}%
\usepackage{algorithmicx}%
\usepackage{algpseudocode}%
\usepackage{listings}%

\newcommand{\1}{\mathbbm{1}}

\newcommand{\R}{\mathbb{R}}
\newcommand{\E}{\mathbb{E}}
\newcommand{\p}{\mathbb{P}}

 \newtheorem{lemma}{Lemma}

\newtheorem{theorem}{Theorem}
\newtheorem{remark}{Remark}%

\begin{document}

\title[General Lotto Games with Scouts: Information versus Strength]{General Lotto Games with Scouts: Information versus Strength}


\author*[1,2]{\fnm{Jan-Tino} \sur{Brethouwer}}\email{J.F.Brethouwer@tudelft.nl}

\author[1,2]{\fnm{Bart} \sur{van Ginkel}}

\author[2]{\fnm{Roy} \sur{Lindelauf}}

\affil[1]{\orgdiv{Institute of Applied Mathematics}, \orgname{Delft University of Technology}, \orgaddress{\city{Delft}, \postcode{2628 CD}, \country{The Netherlands}}}

\affil[2]{\orgdiv{Faculty of Military Sciences}, \orgname{Netherlands Defence Academy}, \orgaddress{\city{Breda}, \postcode{4818 CR}, \country{The Netherlands}}}

\abstract{We introduce General Lotto games with Scouts: a General Lotto game with asymmetric information. There are two players, Red and Blue, who both allocate resources to a field. However, scouting capabilities afford Blue to gain information, with some probability, on the number of Red's resources before allocating his own. We derive optimal strategies for this game in the case of a single field. In addition we provide upper and lower bounds of the value of the game in a multi-stage case with multiple battlefields.
We devise several ways to characterise the influence of information versus strength. We conclude by drawing qualitative insights from these characterisations and the game values, and draw parallels with military practice. }




\maketitle

\section{Introduction}
The Colonel Blotto game in its most basic form was introduced by \cite{borel1921theorie}. The main idea is that two opposing Colonels simultaneously allocate a fixed amount of resources over a finite number of fields. Each field is won by whichever player sends the most resources there, and the goal is to win as many fields as possible. The Colonel Blotto game and its variations and generalisations have been used to model various strategic situations, ranging from division of troops and the placement of missiles (\cite{eckler1972mathematical}) to budget allocation during elections (\cite{washburn2013or,laslier2002distributive}).\\
\\
Borel introduced  two versions of the Colonel Blotto game in his original paper: the discrete Colonel Blotto game, where the number of resources sent to a field has to be an integer, and the continuous Colonel Blotto game, where the number of resources can be any non-negative number.  The first solutions to the continuous Colonel Blotto game were found by \cite{gross1950continuous}. They solved the game completely for 2 fields. They also provided solutions for 3 or more fields, under the condition that both sides have an equal number of resources. The continuous Blotto game was eventually completely solved by \cite{roberson2006colonel}, allowing for any number of fields and asymmetric resource numbers between players. The discrete version turned out to be quite difficult to solve, as even now it is considered very computationally intensive to find solutions. There is a substantial literature on the computational complexity of the discrete Blotto game, see \cite{behnezhad2017faster} and the references therein. Even so, \cite{liang2023discrete} showed that the equilibrium payoff for two-battlefield games is generally the same in the continuous version as in the discrete version.\\ 
\\ 
\cite{hart2008discrete} introduced a more computable variant of Colonel Blotto games, which he called General Lotto games. These games relax the budget constraint on the number of resources allocated to a field, only requiring it to hold in expectation. 
Hart showed how these games could be used to find optimal strategies for symmetric discrete Colonel Blotto games. While General Lotto games were introduced as a tool to study Colonel Blotto games, they have shown to be interesting games in their own right.\\ 
\\
In standard General Lotto games, the resources allocated to each of the fields are revealed simultaneously. In practice, however, this need not always be the case. Intelligence, surveillance and reconnaissance efforts might unveil adversary troop allocations beforehand. A reconnaissance drone or satellite image could identify troops en route to the field, a spy might uncover the number of missiles in a specific location, or a leaked document could detail the budget allocated to a certain state during elections. Knowledge of such events significantly influences the strategies of both players.\\
\\
In this paper, we present a new variation of the General Lotto game that accommodates such information. We consider the following variation of a General Lotto game. Consider two players, Blue and Red. First, Red allocates a number of resources to the field, as done in a standard General Lotto game. Next, Blue may obtain knowledge of Red's resource allocation, i.e., with some probability Blue learns the exact number of resources Red has allocated. Blue can now choose to match this number and win the corresponding field, or allocate no resources to the field. Since it is a General Lotto game, the number of resources allocated by the players has to match their budget in expectation.\\
\\
Our research is inspired by the real-world scenario of a scouting drone spotting enemy troops on their way to a field, which is a problem that was provided to us by the Dutch ministy of Defense. This setting has since then only grown in relevance, with the war in Ukraine using surveillance drones on a large scale, making troop movements more transparent than ever (\cite{gady2023how}). We abstracted this to a problem of information versus strength with fixed budgets.\\
%
%
\\
The goal of this paper is to study the effect of information asymmetry on resource allocation. Furthermore, we want to gain insights in optimal configurations of different types of resources: resources that yield strength (troops) or information (drones).\\
We focus on two main research questions. 
First, we assume that the amount of information and strength is predetermined. 
With the information that is provided, what are the optimal strategies? Specifically, how does Blue use his information efficiently and how does Red change her strategy to deal with Blue possibly gaining information? We provide several insights and guidelines to answer to these questions.\\
\\
The second research question is a generalisation of the first: What if you are allowed to choose your own configuration of information and strength under a given budget? This is the so-called weapons mix problem (\cite{jaiswal2012military}). You are about to face a series of battles, and you have to assemble your troops. Given is a fixed budget to divide between scouting drones which provide information, and weaponized unmanned vehicles which provide firepower. What is the mix of information and strength that maximises expected wins?



\subsubsection*{Relevant literature}

General Lotto games are a valuable framework for exploring information and asymmetry in resource allocation. There are many aspects of the game where there can exist asymmetry: players may have different field valuations, resources could vary in strength between players, or a player might have a strategic advantage, such as in our game. \\
\\
Several extensions of General Lotto Games have been developed to investigate the effects of these asymmetries. For example, in \cite{kovenock2021generalizations} the players have asymmetric valuations for the fields, which are known to both players. Another approach is in \cite{paarporn2022asymmetric}, where the valuation of fields is the same for both players, but only one player knows the exact value, and the other player only knows its distribution.\\
\\
\cite{vu2021colonel} provided solutions for General Lotto games with favouritism, where some of the resources are pre-allocated beforehand, and remaining resources have differing strengths per field and player. The number of resources pre-allocated for each field is fixed.
\cite{chandan2022strategic} builds upon this idea, by including pre-allocation in the strategy. They extend General Lotto games to two-stage games, where in the first phase a fixed budget of resources must be pre-allocated over fields, but the player now has freedom to choose how to divide these resources. In the second stage, these resources are revealed, and the General Lotto game with favouritism is played.\\
\\
An alternate extension of General Lotto games are Winner-Take-All games, as studied by \cite{alpern2017winner}. Where General Lotto games have two players competing, Winner-Take-All games allow for any number of players. Each player is tasked with selecting a distribution following predefined constraints. Subsequently, scores are sampled from these distributions, and the player achieving the highest score wins. Notably, two player Winner-Take-All games with constraints on the expectation of the distribution are equivalent to General Lotto games. 

\subsubsection*{Main contributions and results}
We first study General Lotto games with Scouts on a single field. We show that based on the amount of resources and the probability that the field is revealed, the analysis of this game can be split into three different cases. We provide optimal strategies for all three. We also study a multi-stage version of the game with multiple fields, which all have their own value and own probability to be revealed. For this version, we provide upper and lower bounds on the value of the game. These bounds are tight in some settings. We devise several ways to compare the value of information and resources. We conclude with several qualitative insights. 
  
\subsubsection*{Overview of the paper}
We first introduce the game model and its notation in Section \ref{sec:model}. We focus on single field General Lotto games with Scouts in Section \ref{sec:singlefield}, where we provide optimal strategies for both players.
In Section \ref{sec:multifield} we move on to a multi-stage version with multiple fields and provide upper and lower bounds on the value of this game. In Section \ref{sec:Measuring}, we introduce methods to measure information versus strength. We discuss the insights gained about the value of information in Section \ref{sec:interpretation}. We conclude and discuss future work in Section \ref{sec:conclusion}.

\section{Model, notation and assumptions} \label{sec:model}
As described in the introduction, General Lotto games as introduced by \cite{hart2008discrete} are a relaxed version of Colonel Blotto games, in the sense that the resource constraint should only hold in expectation, rather than with probability $1$. We want to study a version of the General Lotto game where one player receives information on the resource configuration chosen by the other player. This act of receiving information is symbolised by the idea of sending a scout. The game is therefore denoted as $GL\text{-}S$, meaning General Lotto with Scouts. In this section we first formally introduce the General Lotto game and its optimal strategies. Afterwards we generalise it to General Lotto with Scouts on a single field, which exhibits insightful properties. In later sections we analyse a multi-stage version with multiple fields, which will be introduced in the relevant section. 

\subsection{General Lotto}
Before introducing General Lotto with Scouts, we first define the General Lotto game. This game involves two players, Red and Blue, competing on a single field. Blue is endowed with resource budget $B>0$, and Red possesses a resource budget $R>0$. Red chooses as her strategy a (distribution of a) nonnegative random variable X with $\E(X) = R$, Blue chooses as his strategy a (distribution of a) nonnegative random variable $Y$ with $E(Y) = B$.\\
\\
To determine the winner, two numbers, $r$ and $b$, representing the amount of resources allocated to the field by Red and Blue respectively, are sampled from their strategies. The player with the highest amount of allocated resources wins the game, with ties being resolved in favor of Blue. I.e. the game is constant-sum with the following pay-off for Blue:
\begin{align*}
    P(b,r) = \1(b\geq r).
\end{align*}
The way in which ties are broken is chosen to avoid technical complications and does not influence the value of the game. We define the value of the game as the probability Blue wins: 
\begin{align*}
    V = P(X \leq Y )
\end{align*}
This game was solved by \cite{hart2008discrete}. He proved the value of the game and the optimal strategies were as in Theorem \ref{thm:Lotto_hart}.\\

\begin{theorem}[Hart]  \label{thm:Lotto_hart}
Let $R \geq B > 0$. The value $V$ of the General Lotto game is
\begin{align*}
    V = \frac{B}{2R}
\end{align*}
and the unique optimal strategies are $X^*$ for Red and $Y^*$ for Blue, where
\begin{align*}
    X^* &\sim \mathbb{U}[0,2R] \\
    Y^* &\sim \mathbb{U}[0,2R] \text{ with prob. } \frac{B}{R} \text{, else } 0
\end{align*} 
\end{theorem}
\subsection{General Lotto with Scouts}
We now consider General Lotto games with Scouts, denoted as $GL\text{-}S(B,R,u)$. Similar to the standard General Lotto game, $GL-S$ involves the two players Blue and Red, with resource budgets $B,R>0$, competing on a single field. However, in this variant, Blue obtains information insight into Red's resource allocation with a probability $u\in [0,1]$. This influences how both players approach the game.\\ 
\\
To describe the feasible strategies, we begin with an assumption and an observation.
\begin{itemize}
    \item First, we assume that both players choose their strategies before knowing whether the number of Red resources will be revealed. This means that the strategy of Blue consists of two parts: actions in case the field is revealed, and actions for when it is not. The expectation over these two together has to match Blue's total resources $B$. 
    \item Second, note that if the number of resources allocated by Red is revealed, Blue either allocates the exact same amount of resources (Call) or allocates no resources at all (Fold). It is clear that any other course of action is less beneficial for Blue, i.e. is dominated by this strategy.
\end{itemize}
A strategy for Red, denoted by $\sigma_R = (X)$, is a random variable $X$ on $[0,\infty)$ that satisfies:
\begin{align*}
    \E (X) \leq R.
\end{align*}
The space of all Reds strategies is denoted by $\Sigma_R$. We denote the number of resources allocated by Red by $T_R(X)$.\\
\\
A strategy for Blue, denoted by $\sigma_B = (t,Z)$, consists of a function $t: [0,\infty) \to [0,1]$ and a random variable $Z$ on $[0,\infty)$. Here $t(x)$ is defined as the probability with which Blue calls if it is revealed that Red has allocated $x$ resources and $Z$ determines the number of resources that Blue allocates if the field is not revealed. We denote the number of resources allocated by Blue by $T_B(t,Z,X)$. Now we see that
\begin{align*}
T_B(t,Z,X) =
    \begin{cases}
    X \text{ with prob. $t(X)$, else 0\hspace{5mm}}  & \text{if the field is revealed} \\
    Z           ,& \text{if the field is not revealed }
\end{cases}
\end{align*}
Given that Red plays strategy $\sigma_R = (X)$ and Blue plays strategy $\sigma_B = (t,Z)$, we see that
\begin{align*}
    \E T_B(t,Z,X) = u\E \left(t(X)X \right) +  (1-u)\E (Z).
\end{align*}
Since the Blue resource constraint must hold for any strategy of Red, we require that
\begin{align*}
   u\E \left(t(X)X \right) +  (1-u)\E (Z) = \E T_B(t,Z,X) \leq B
\end{align*}
for all random variables $X$ on $[0,\infty)$ that satisfy $\E (X)\leq R$. The space of all such strategies of $B$ is denoted by $\Sigma_B$.\\
\\
The value of the game is the expected payoff to Blue.
Given that Blue plays strategy $\sigma_B = (t,Z)$ and Red plays strategy $\sigma_R = (X)$, the value of the game $V$ is given by
\begin{align*}
     V = P(\sigma_B,\sigma_R) = u \E(t(X)) + (1-u)\p(Z \geq X).
\end{align*}
Blue want to find a strategy which maximises the value, and Red want to find a strategy which minimises the value.
\section{General Lotto with Scouts: Single Field}\label{sec:singlefield}
In the last section we introduced the $GL\text{-}S$ game. In this section we find its value and optimal strategies. The main result is given in Theorem~\ref{thm:GL-S} at the end of the section.\\
\\
Based on the resource ratio $B/R$, we split the analysis into the following three cases:
\begin{itemize}
    \item Case 1: $GL\text{-}S$ with $1 \leq B/R $
    \item Case 2: $GL\text{-}S$ with $u \leq B/R \leq 1$
    \item Case 3: $GL\text{-}S$ with $B/R \leq u$
\end{itemize}
The cases have different optimal strategies and are therefore treated separately. At the end of the section, the results are combined into Theorem~\ref{thm:GL-S}. Every time we propose two strategies and show that they both guarantee the same payoff $V$. Thus these strategies must be optimal and the value of the game is $V$.
 
\subsection{Case 1: \texorpdfstring{$GL\text{-}S$}{GL-S} with \texorpdfstring{$1 \leq B/R $}{1 < B/R}}
We introduce two new parameters $q(u)$ and $C(u)$:
\begin{align*}
    q(u) &= \frac{1-u}{\left( \frac{B}{R}\right) -u}\\
    C(u) &= \frac{R}{q(u)} = R \left(\frac{\frac{B}{R}-u}{1-u}\right) = B + u \frac{B-R}{1-u}
\end{align*}
Since $B \geq R$, it follows that $0\leq q(u) \leq 1$ and $C(u)\geq B \geq R$.\\
\\
We show that $\sigma_R^* = (X^*)$ and $\sigma_R^* = (t^*,Z^*)$ are optimal strategies, where:
\begin{align*}
    X^* &\sim \mathbb{U}[0,2C(u)] \text{ with prob. } q(u) \text{, else } 0\\
    t^* &= 1 \quad \forall x \\
    Z^* &\sim \mathbb{U}[0,2C(u)]
\end{align*}
$\mathbb{U}[a,b]$ denotes the uniform distribution with support $[a,b]$. This strategy consists of Blue always replicating Reds allocation if the field is revealed, or Blue plays a uniform distribution if the field is not revealed. Red plays the same uniform distribution with probability $q(u)$, and 0 with probability $1-q(u)$.\\

\begin{lemma} \label{lem:case1}
These strategies are optimal strategies in the case $\frac{B}{R} \geq 1$, and the value of the game equals $V=1-\frac{(1-u)^2}{2(\frac{B}{R}-u)}$.
\end{lemma}
\begin{proof}
First we show that the resource constraints are satisfied:
\begin{eqnarray*}
    \E T_R(X^*) &=& \E (X^*) = q(u) \cdot C(u) + (1-q(u)) \cdot 0 = q(u) \cdot \frac{R}{q(u)} = R \\
    \E T_B(t^*,Z^*,X) &=& u\E \left(t^*(X) X \right) +  (1-u)\E (Z^*) = u \cdot R + (1-u) \cdot C(u) =B,
\end{eqnarray*}
where $X$ is an arbitrary non-negative random variable such that $\E X = R$.\\
Now it suffices to prove the following two claims.\\
\\
\textbf{Claim 1: given that Red plays $\sigma_R^*$, the payoff is at most $V$.}\\
If Red plays $\sigma_R^*$, the payoff for any $\sigma_B$ is:
\begin{align*}
    P(\sigma_B,\sigma_R^*) 
    &= u \E(t(X^*)) + (1-u)\E \left( \1[Z \geq X^*] \right)
\end{align*}
 We can rewrite $\E \left( \1[Z \geq X^*]\right)$ by conditioning on $Z$:
\begin{equation*}
    \E \left( \1[Z \geq X^*]\right)
    = \E \left( \E \left( \1[Z \geq X^*]|Z \right) \right)\\
    = 1-q(u) + q(u)\E \left( \E \left( \1[Z \geq U ]|Z \right) \right)
\end{equation*}
    Where $U \sim \mathbb{U}[0,2C(u)]$. We can use this to rewrite as follows:
    \begin{align*}
        \E \left( \1[Z \geq U ]|Z \right) &= \left( \frac{Z}{2C(u)}\right)\1[2C(u) \geq Z] + 1 \cdot \1[2C(u)<Z]\\
        &= \left( \frac{Z}{2C(u)}\right)\left(1-\1[2C(u) < Z]\right) + 1 \cdot \1[2C(u)<Z]\\
        &= \left( \frac{Z}{2C(u)}\right)- \left( \frac{Z}{2C(u)}-1\right) \cdot \1[2C(u)<Z]
    \end{align*}
    Plugging this back in gives:
\begin{align*}
    \E \left( \1[Z \geq X^*]\right)
    &= 1-q(u) + q(u)\left(\E \left( \frac{Z}{2C(u)}\right) -   \E \left( \left( \frac{Z}{2C(u)}-1\right) \cdot \1[2C(u)<Z] \right) \right)\\
    &= 1-q(u) + q(u)\left(\frac{\E(Z)}{2C(u)} - \E \left( \left( \frac{Z}{2C(u)}-1\right) \cdot \1[2C(u)<Z] \right) \right)\\
    &= 1-q(u) \left(1-\frac{\E(Z)}{2C(u)}\right) - q(u)\E \left( \left( \frac{Z}{2C(u)}-1\right) \cdot \1[2C(u)<Z] \right)\\
    &\leq 1-q(u) \left(1-\frac{\E(Z)}{2C(u)}\right)
\end{align*}
Using this equation, we can rewrite the payoff as:
\begin{align*}
    P(\sigma_B,\sigma_R^*) 
    &\leq u\E(t(X^*)) + (1-u)\left(1-q(u) \left(1-\frac{\E(Z)}{2C(u)}\right)\right)\\
    &= u\E(t(X^*)) + (1-u)(1-q(u))+\frac{q(u)}{2C(u)}\left( B- u\E \left(t(X^*) X^* \right)\right)\\
    &= u\E\left(t(X^*)\left(1-\frac{q(u)X^*}{2C(u)}\right)\right)+ (1-u)(1-q(u))+\frac{q(u)B}{2C(u)}
\end{align*}
Since $0\leq \left(1-\frac{q(u)X^*}{2C(u)}\right)$, this expression is maximised by taking $t(x)=1 \forall x$. This yields
\begin{align*}
    P(\sigma_B,\sigma_R^*) 
    &\leq u\E\left(1-\frac{q(u)X^*}{2C(u)}\right)+ (1-u)(1-q(u))+\frac{q(u)B}{2C(u)}\\
    &= 1 + u - u -q(u)(1-u) + \frac{q(u)}{2C(u)}\left( B-uR\right)\\
    &= 1-q(u)(1-u) +\frac{q(u)}{2}  \frac{B-uR}{C(u)}\\
    &= 1-\frac{1}{2}q(u)(1-u) = 1-\frac{(1-u)^2}{2(\frac{B}{R}-u)} =V
\end{align*}
This proves Claim 1.\\
\\
\textbf{Claim 2: given Blue plays $\sigma_B^*$, the payoff is at least $V$.}\\
If Blue uses the strategy $\sigma_B^*$, the payoff for any $\sigma_R$ is:
\begin{align*}
    P(\sigma_B^*,\sigma_R) 
    &= u \E(t^*(X)) + (1-u)\E \left( \1[Z^* \geq X] \right) \\
    &= u \E(1) + (1-u)\E \left( \1[Z^* \geq X] \right) \\
    &= u + (1-u) \E \left( \1[Z^* \geq X] \right)
    \end{align*}
    We can rewrite $\E \left( \1[Z^* \geq X] \right)$ by conditioning on $X$:
    \begin{align*}
        \E \left( \1[Z^* \geq X] \right)
        &= \E \left( \E \left( \1[Z^* \geq X] | X\right) \right)
         \end{align*}
        Where $Z^* \sim \mathbb{U}[0,2C(u)]$. We can use this to rewrite as follows:
        \begin{align*}
        \E \left( \E \left( \1[Z^* \geq X] | X\right) \right)
        &= \E \left(\left( 1-\frac{X}{2C(u)}\right)\1[2C(u) \geq X] + 0 \cdot \1[2C(u)<X]  \right)\\
        &= \E \left(\left( 1-\frac{X}{2C(u)}\right)\left( 1- \1[2C(u)<X] \right) \right)\\
        &= \E \left( 1-\frac{X}{2C(u)}\right) + \E \left( \1[2C(u)<X]\left( \frac{X}{2C(u)}-1\right) \right)\\
        &= 1 - \frac{R}{2C(u)} + \E \left( \1[2C(u)<X]\left( \frac{X}{2C(u)}-1\right) \right)\\
        &\geq 1 - \frac{R}{2C(u)} = 1 - \frac{q(u)}{2}
    \end{align*}
The last step follows from the fact that the expectation on the RHS is non-negative.
Plugging this into the original equation for payoff gives:
\begin{align*}
     P(\sigma_B^*,\sigma_R) 
     &\geq  u + (1-u) \left(1 - \frac{q(u)}{2}\right) = 1 -(1-u)\frac{q(u)}{2}\\
     &= 1-\frac{(1-u)^2}{2(\frac{B}{R}-u)} = V
\end{align*}
This proves Claim 2 and therefore completes the proof.
\end{proof}

\subsection{Case 2: \texorpdfstring{$GL\text{-}S$}{GL-S} with \texorpdfstring{$u \leq B/R \leq 1$}{u < B/R < 1}}
This is the most interesting of the three scenario's. Blue is outnumbered, but can use information to be competitive. 
We introduce the following parameter:
\begin{align*}
    p(u) = \frac{\frac{B}{R}-u}{1-u} = \frac{1}{q(u)}
\end{align*}
Since $u\leq \frac{B}{R} \leq 1$, it follows that $0 \leq p(u) \leq 1$.\\
\\
We show that $\sigma_R^* = (X^*)$ and $\sigma_R^* = (t^*,Z^*)$ are optimal strategies, where:
\begin{align*}
    X^* &\sim \mathbb{U}[0,2R]\\
    t^* &= 1 \quad \forall x \\
    Z^* &\sim \mathbb{U}[0,2R] \text{ with prob. } p(u) \text{, else } 0
\end{align*}
This strategy consists of Blue always replicating Reds allocation if the field is revealed, or Blue plays a uniform distribution with probability $p(u)$ if the field is not revealed, and 0 with probability $1-p(u)$. Red always plays the same uniform distribution, which is equivalent to her strategy in General Lotto games without scouts.\\

\begin{lemma} \label{lem:case2}
    These strategies are optimal in the case $u\leq \frac{B}{R} \leq 1$, and the value of the game equals $V=\frac{1}{2}\left(u+\frac{B}{R}\right)$.
\end{lemma}
\begin{proof}
First we show that the resource constraints are satisfied:
\begin{eqnarray*}
    \E T_R(X^*) &=& \E (X^*) = R \\
    \E T_B(t^*,Z^*,X) &=& u\E \left(t^*(X) X \right) +  (1-u)\E (Z^*) = u R + (1-u) p(u) R = B,
\end{eqnarray*}
where $X$ is an arbitrary non-negative random variable such that $\E X = R$.\\
Now it suffices to prove the following two claims.\\
\\
\textbf{Claim 1: given that Red plays $\sigma_R^*$, the payoff is at most $V$.}\\
If Red plays $\sigma_R^*$, the payoff for any $\sigma_B$ is:
\begin{align*}
    P(\sigma_B,\sigma_R^*) 
    &= u \E(t(X^*)) + (1-u)\E \left( \1[Z \geq X^*] \right)
\end{align*}
 We can rewrite $\E \left( \1[Z \geq X^*]\right)$ by conditioning on $Z$:
\begin{align*}
    \E \left( \1[Z \geq X^*]\right)
    &= \E \left( \E \left( \1[Z \geq X^*]|Z \right) \right)
    \end{align*}
    Where $X^* \sim \mathbb{U}[0,2R]$. We can use this to rewrite as follows:
    \begin{align*}
    \E \left( \E \left( \1[Z \geq X^*]|Z \right) \right)
    &= \E \left(\left( \frac{Z}{2R}\right)\1[2R \geq Z] + 1 \cdot \1[2R<Z]  \right)\\
    &= \E \left(\left( \frac{Z}{2R}\right)\left(1-\1[2R < Z]\right) + 1 \cdot \1[2R<Z]  \right)\\
    &= \E \left( \frac{Z}{2R}\right) -   \E \left( \left( \frac{Z}{2R}-1\right) \cdot \1[2R<Z] \right)\\
    &= \frac{\E(Z)}{2R} - \E \left( \left( \frac{Z}{2R}-1\right) \cdot \1[2R<Z] \right)\\
    &\leq \frac{\E(Z)}{2R}
\end{align*}
Using this equation, we can rewrite the payoff as:
\begin{align*}
    P(\sigma_B,\sigma_R^*) 
    &\leq u\E(t(X^*)) + (1-u)\left(\frac{\E(Z)}{2R}\right)
\end{align*}    
If we rewrite the expectation requirement, we get:
\begin{align*}
    (1-u)\E(Z) = B - u\E \left(t(X) X \right)
\end{align*}
Plugging this in gives:
\begin{align*}
     P(\sigma_B,\sigma_R^*) 
    &\leq u\E(t(X^*)) + \frac{1}{2R} \left( B - u\E \left(t(X^*) X^* \right) \right)\\
    &= \frac{B}{2R} + u\E\left(t(X^*)\left(1-\frac{X^*}{2R}\right)\right)
\end{align*}
Since $0\leq \left(1-\frac{X^*}{2R}\right)$, the payoff is optimised by taking $t(\cdot)$ as large as possible. It it therefore maximised by taking $t(x)=1 \forall x$, which gives:
\begin{align*}
    P(\sigma_B,\sigma_R^*) 
    &\leq \frac{B}{2R} + u\E \left(1\left(1-\frac{X^*}{2R}\right)\right) = \frac{B}{2R} + u\left(1-\frac{R}{2R} \right)\\
    &= \frac{1}{2}\left( \frac{B}{R} +u \right) = V
\end{align*}
This proves Claim 1.\\
\\
\textbf{Claim 2: given that Blue plays $\sigma_B^*$, the payoff is at least $V$.}\\
If Blue uses the strategy $\sigma_B^*$, the payoff for any $\sigma_R$ is:
\begin{align*}
    P(\sigma_B^*,\sigma_R) 
    &= u \E(t^*(X)) + (1-u)\E \left( \1[Z^* \geq X] \right) \\
    &= u \E(1) + (1-u)\E \left( \1[Z^* \geq X] \right) \\
    &= u + (1-u) \E \left( \1[Z^* \geq X] \right)
    \end{align*}
    We can rewrite $\E \left( \1[Z^* \geq X] \right)$ by conditioning on $X$:
    \begin{align*}
        \E \left( \1[Z^* \geq X] \right)
        &= \E \left( \E \left( \1[Z^* \geq X] | X\right) \right)\\
        &= \E \left( \E \left( (1-p(u)) \cdot 0 + p(u)\1[U \geq X] | X\right) \right)\\
        &= p(u) \E \left(\E \left(\1[U \geq X] | X\right)\right)
    \end{align*}
    Where $U \sim \mathbb{U}[0,2R]$. We can use this to rewrite as follows:
    \begin{align*}
    \E \left(\E \left(\1[U \geq X] | X\right)\right)
        &= \E \left(\left( 1-\frac{X}{2R}\right)\1[2R \geq X] + 0 \cdot \1[2R<X]  \right)\\
        &= \E \left(\left( 1-\frac{X}{2R}\right)\left( 1- \1[2R<X] \right) \right)\\
        &= \E \left( 1-\frac{X}{2R}\right) +  \E \left( \1[2R<X]\left( \frac{X}{2R}-1\right) \right)\\
        &= \left(1 - \frac{R}{2R}\right) + \E \left( \1[2R<X]\left( \frac{X}{2R}-1\right) \right)\\
        &= \frac{1}{2} + \E \left( \1[2R<X]\left( \frac{X}{2R}-1\right) \right)\\
        &\geq \frac{1}{2}
    \end{align*}
Here the last step follows from the fact that the expectation on the RHS is non-negative.
Plugging this into the original equation for the game payoff yields
\begin{align*}
     P(\sigma_B^*,\sigma_R) 
     &\geq  u + (1-u) \left(\frac{p(u)}{2}\right) = u +\frac{\frac{B}{R}-u}{2}\\
     &= \frac{1}{2}\left(u+\frac{B}{R}\right) = V
\end{align*}
This proves Claim 2 and therefore completes the proof.
\end{proof}


\subsection{Case 3: \texorpdfstring{$GL\text{-}S$}{GL-S} with \texorpdfstring{$B/R \leq u$}{B/R < u}}
In this case Blue is outnumbered, but has a significant amount of information. The optimal strategies of this case are surprising, as Red will do away with any randomness in her strategy.\\
\\
We show that $\sigma_R^* = (X^*)$ and $\sigma_R^* = (t^*,Z^*)$ are optimal strategies, where:
\begin{align*}
    X^* &= R\\
    t^* &= \frac{B}{uR} \quad \forall x \\
    Z^* &= 0
\end{align*}
This strategy consists of Blue replicating Reds allocation with probability $\frac{B}{uR}$ whenever the field is revealed, independently of the number of resources spotted. If the field is not revealed, Blue allocates no resources. Red always allocates $R$ resources, effectively removing any random element from her strategy, and thus removing the advantage Blue has with his scouts.\\

\begin{lemma} \label{lem:case3}
    These strategies are optimal strategies in the case $u\geq \frac{B}{R}$, and the value of the game is $V=\frac{B}{R}$: 
\end{lemma}
\begin{proof}
First we show that the resource constraints are satisfied:
\begin{eqnarray*}
    \E T_R(X^*) &=& \E (X^*) = R \\
    \E T_B(t^*,Z^*,X) &=& u\E \left(t^*(X) X \right) +  (1-u)\E (Z^*) u \left(\frac{B}{uR} \right) \E \left( X\right) + (1-u) \cdot 0 = B,
\end{eqnarray*}
where $X$ is an arbitrary non-negative random variable such that $\E X = R$.\\
Now it suffices to prove the following two claims.\\
\\
\textbf{Claim 1: given that Red plays $\sigma_R^*$, the payoff is at most $V$.}\\
If Red uses the strategy $\sigma_R^*$, the payoff for any $\sigma_B$ is:
\begin{align*}
    P(\sigma_B,\sigma_R^*) 
    &= u \E(t(X^*)) + (1-u)\E \left( \1[Z \geq X^*] \right)\\
    &= u \E(t(X^*)) + (1-u)\E \left( \1[Z \geq R] \right)\\
    &= \frac{1}{R} \left( u \E(R \cdot t(X^*)) + (1-u)\E \left( R \cdot \1[Z \geq R] \right)\right)
\end{align*}
Note that: 
\begin{align*}
    \E(R \cdot t(X^*)) = \E(t(X^*)X^*)
\end{align*}
and
\begin{align*}
    \E \left( R \cdot \1[Z \geq R] \right) \leq \E \left( Z \cdot \1[Z \geq R] \right) \leq \E(Z)  
\end{align*}
Plugging these into the equation for the payoff gives:
\begin{align*}
    P(\sigma_B,\sigma_R^*) 
    &\leq \frac{1}{R} \left( u \E \left( t(X^*)X^* \right) + (1-u)\E \left( Z \right)\right)\\
    &= \frac{1}{R} \left(B\right) = \frac{B}{R} = V,
\end{align*}
where the last step follows from the expectation requirement for Blue.\\
This proves Claim 1.\\
\\
\textbf{Claim 2: given that Blue plays $\sigma_B^*$, the payoff is at least $V$.}\\
If Blue uses the strategy $\sigma_B^*$, the payoff for any $\sigma_R$ is:
\begin{align*}
    P(\sigma_B^*,\sigma_R) 
    &= u \E(t^*(X)) + (1-u)\E \left( \1[Z^* \geq X] \right) \\
    &= u \left(\frac{B}{uR}\right) + (1-u) \cdot P[X=0] \\
    &\geq \frac{B}{R} = V.
    \end{align*}
This proves Claim 2 and therefore completes the proof.
\end{proof}
\noindent
As the payoff in this case no longer depends on the $u$, it follows that increasing the $u$ past $\frac{B}{R}$ no longer increases payoff.

\subsection{Solution of \texorpdfstring{$GL\text{-}S$}{GL-S}}
We combine the results of Lemma~\ref{lem:case1}, \ref{lem:case2} and \ref{lem:case3} to form our main result:\\

\begin{theorem}  \label{thm:GL-S}
    In $GL\text{-}S$, it is optimal for Blue to play $\sigma_B^*=(t^*,Z^*)$ and optimal for Red to play $\sigma_R^*=(X^*)$, where $X^*,t^*$ and $Z^*$ are chosen according to Table~\ref{tab:Optimal_Strategies}.\\ 
\bgroup
\def\arraystretch{1.5}
\begin{table}[!ht]
    \centering
    \begin{tabular}{c ||c |c |c}
         &  $X^*$ & $t^*(\cdot)$ & $Z^*$\\ \hhline{=||=|=|=}
         $\frac{B}{R} \leq u$        & R & $\frac{B}{uR}$ & 0\\[1ex] \hline
         $u \leq \frac{B}{R} \leq 1$ & $\mathbb{U}[0,2R]$ &1 & $
             \mathbb{U}[0,2R] \text{ w.p. } p(u), \text{ else } 0 $ \\[1ex] \hline
         $1 \leq \frac{B}{R}$        & $\mathbb{U}[0,2C(u)] \text{ w.p. } q(u), \text{ else } 0$ &1 & $\mathbb{U}[0,2C(u)]$
    \end{tabular}
    \caption{Optimal strategy of $GL\text{-}S$.}
    \label{tab:Optimal_Strategies}
\end{table}\\
\egroup
The value $V$ of $GL\text{-}S$ is therefore:
\begin{align*}
    V = P(\sigma_B^*,\sigma_R^*) = 
    \begin{cases}
   \frac{B}{R},                             & \text{if } \frac{B}{R} \leq u,\\
   \frac{u+\frac{B}{R}}{2},                 & \text{if } u \leq \frac{B}{R} \leq 1\\
    1-\frac{(1-u)^2}{2(\frac{B}{R}-u)},     & \text{if } 1 \leq \frac{B}{R},
\end{cases}
\end{align*}
\end{theorem}
\noindent
The value of $GL\text{-}S$ depends both $u$ and the ratio $B/R$. We illustrate this function in three different figures.
The detection chance $u$ is fixed in Figure~\ref{fig:single_field_value}, the ratio $B/R$ is fixed in Figure~\ref{fig:Lotto} and the value of the game is portrayed as a heatmap in Figure~\ref{fig:LottoPayoffHeatmap}.\\

\begin{figure}[!ht]
    \centering
    \includegraphics[width=\textwidth]{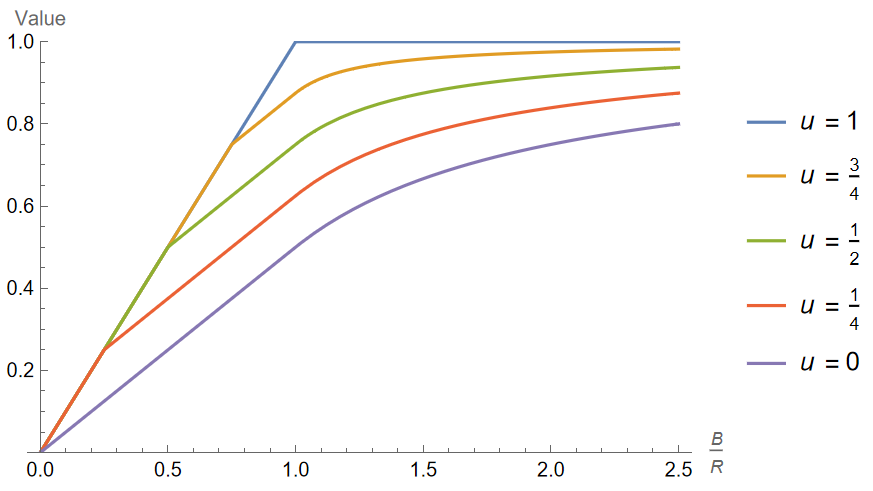}
    \caption{The value of $GL\text{-}S$ for fixed detection probability $u$ as a function of $B/R$.}
    \label{fig:single_field_value}
\end{figure}

\begin{figure}[!ht]
    \centering
    \includegraphics[width=\textwidth]{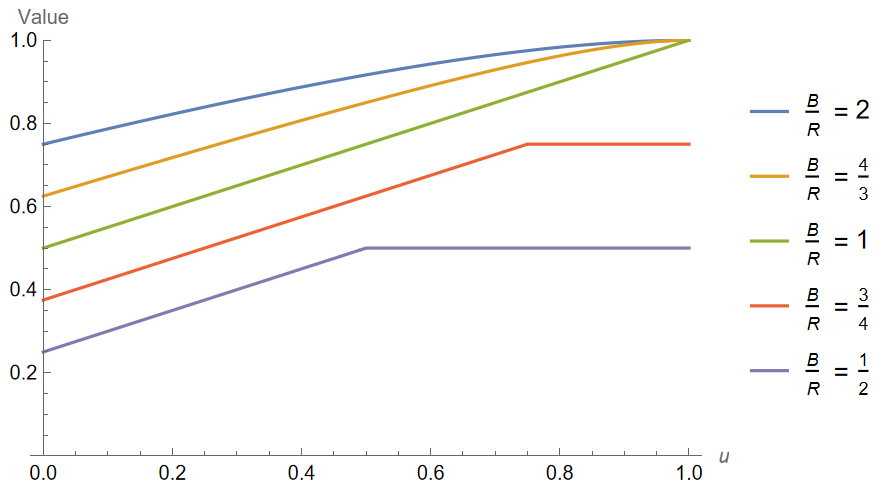}
    \caption{The value of $GL\text{-}S$ for fixed $B/R$ ratios as a function of the detection probability $u$.}
    \label{fig:Lotto}
\end{figure}

\begin{figure}[!ht]
    \centering
    \includegraphics[width=\textwidth]{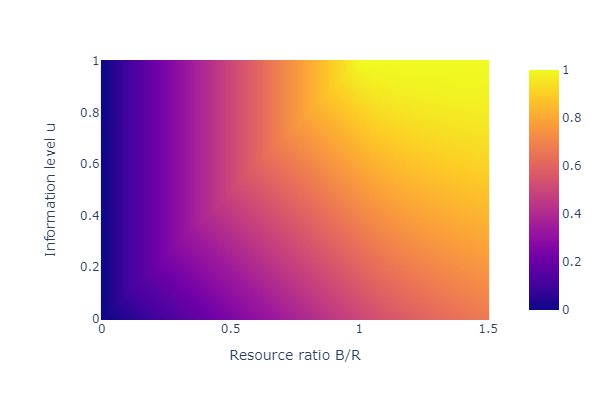}
    \caption{A heatmap of the value of $GL\text{-}S$, with the resource ratio $B/R$ on the x-axis and the detection probability $u$ on the y-axis.}
    \label{fig:LottoPayoffHeatmap}
\end{figure}

\noindent
Some interesting insights follow from this optimal solution to the $GL\text{-}S$ game. The first is that $t(\cdot)$ is always maximised, meaning that Blue calls as many fields as he can, given his budget. Calling a field is always a more efficient way of using resources, then using resources on a field that was not revealed.\\
\\
The second insight follows from the question whether Red should distribute her resources more or less evenly compared to a standard Lotto game.
%
The answer is both, but in different cases. If Red heavily outnumbers Blue, i.e. $\tfrac{B}{R}<u$, she plays as evenly as possible, sending exactly $R$ resources every time. However, if Blue outnumbers Red, Red plays riskier than in the standard General Lotto game, allowing higher number of resources per round, but also sending 0 resources more often. This also means that, as long as $u<1$, Red always has a chance to win, even if heavily outnumbered.\\
%
%
\\
\begin{remark}[Comparison with known results when $u=0$]\label{rmk:sf}
By setting the detection probability $u$ in $GL\text{-}S$ to $0$, we retrieve a General Lotto game of which the value should be consistent with earlier work. The value of this game is $\tfrac{B}{2R}$ whenever $B\leq R$ and $1 - \frac{R}{2B}$ otherwise. This is consistent with Corollary 1 of \citet{kovenock2021generalizations}. Also the optimal strategies found there coincide with the optimal strategies that we use. It is also consistent with \cite{hart2008discrete} and \cite{sahuguet2006campaign} who showed this result earlier.\\
\end{remark}

\begin{remark}[Comparison with known results when $u=1$]
By setting $u=1$ we obtain a game that is in some respects similar to the pre-allocation game in \citet{chandan2022strategic}. Namely, when $u=1$ Blue always sees what Red does and in that sense Red `pre-allocates'. There are two major differences. The first is that in our case Red does not get the opportunity to allocate any additional resources and in that sense forcibly pre-allocates all her resources. The second difference is that the pre-allocation in \citet{chandan2022strategic} is performed `Blotto-style', i.e. such that the sum of the pre-allocated resources equals a fixed constant. In our case it should only in expectation equal the total resources.
\end{remark}

\section{General Lotto with Scouts: Multistage} \label{sec:multifield}
In this section we consider a multistage version $GL\text{-}MS$ of $GL\text{-}S$ where Blue and Red first divide their total resources over the battlefield and then play $GL\text{-}S$ on each battlefield. The latter step was solved in Section~\ref{sec:singlefield}, so we know the value of the games on the individual battle fields. Therefore, this multistage version of the game can be considered a Colonel Blotto game with a complicated payoff function given by those game values. First we introduce the game and discuss under which conditions the game value exists. Subsequently, we provide strategies for Blue and Red, which provide upper and lower bounds on the value of the multistage game. Lastly, we specify under which conditions these bounds coincide, in which case these strategies must be optimal.

\subsection{The multistage game}
We start out with the same two players with total resources $B,R>0$, respectively. 
Next, let there be $n$ fields, where $n$ is an integer larger than one.
Each field is endowed with a worth $w_i>0$ and a detection probability $u_i \in [0,1]$. We assume that $\sum_i w_i = 1$. We denote the Multistage General Lotto game with Scouts with these parameters by $GL\text{-}MS(B,R,\{u_i\},\{w_i\})$.\\
\\
The game is played in two stages.
\begin{itemize}
    \item[i)] Red chooses a distribution $\underline R=(R_1,..,R_n)$ of her total resources over the fields such that $\sum_i R_i = R$. Simultaneously Blue chooses a distribution $\underline B=(B_1,..,B_n)$ of his total resources over the fields such that $\sum_i B_i = B$.
    \item[ii)] On each field an $GL\text{-}S(B_i,R_i,u_i)$ game is played with worth $w_i$.
\end{itemize}
The total payoff for Blue is the sum of the individual payoffs.\\
\\
Recall from Section~\ref{sec:singlefield} that given the resources $B_i, R_i$ at field $i$ and given its worth $w_i$ and the detection probability $u_i$, the payoff at field $i$ equals $\phi_i(B_i/R_i)$ where
\begin{equation*}
    \phi_i(x) = w_i
    \begin{cases}
        x & x \leq u_i \\
        \frac{1}{2}(x+u_i) & u_i \leq x \leq 1 \\
        1 - \frac{(1-u)^2}{2(x-u)} & x \geq 1
    \end{cases}
\end{equation*}
Note that the field worth $w_i$ was not present before, but is included in $\phi_i$ here. Since for $n=1$ there is just one field with worth $1$, this is consistent with before.\\
\\
Therefore the first stage of the game has payoff function 
\begin{equation*}
    H(\underline B, \underline R) = \sum_{i=1}^n \phi_i\left(\frac{B_i}{R_i}\right).
\end{equation*}
Our goal is to solve this game with the constraints
\begin{equation*}
    \sum_{i=1}^n B_i = B, \hspace{1cm} \sum_{i=1}^n R_i = R.
\end{equation*}

\subsection{Existence of Game Value}
There are three things that make it difficult to derive the game value. First, even though  $H$ is concave in $B_1,..,B_n$, it is not convex in $R_1,..,R_n$. Second, $H$ is not everywhere differentiable, because of the sharp turn in the $\phi_i$'s at $x=u_i$. Third, $H$ ill defined at $B_i=R_i = 0$ for any $i$ and cannot be continuously extended there.\\
\\
Without the singularity at (0,0), the payoff function is continuous and we can guarantee the existence of the game value. Allocating zero resources to a field is both theoretically and in practice unlikely to be an optimal strategy. It is possible to remove this singularity from the game in the following way: we fix $\epsilon >0$ and require that $B_i,R_i \geq \epsilon$ for all $i$. This creates a new, slightly more restricted game. As a result of this restriction $H$ is continuous on the domain of strategies while the domain is still compact. By a classical result of \cite{glicksberg1952further}, there must exists a Nash equilibrium for the corresponding game.\\ 
\\
Now that we have discussed existence, we want to study the exact value of the game.
To bound the value of the game, we provide strategies for Red And Blue and prove that these strategies guarantee certain upper and lower bounds. Whenever those bounds coincide, we know the exact value of the game.\\
\\
To find a good guess for optimal strategies for Red and blue, we explore the case where $H$ is continuously differentiable, even though it is not. If the $\phi$'s were $C^1$, then any Nash equilibrium would have to satisfy the equations in Lemma \ref{lemma:minimax}, with the provided optimal strategies. We use these strategies to provide bounds on the original game.\\

\begin{lemma}\label{lemma:minimax}
Let $\phi_i:[0,\infty)\rightarrow \R$, $i=1,..,n$, be $C^1$ and increasing and let $B,R>0$. Suppose that $B_1,..,B_n$, $R_1,..,R_n$ satisfy
\begin{equation}\label{eq:minimax}
    \sum_{i=1}^n \phi_i\left(\frac{B_i}{R_i}\right) = \inf_{R'_1,..,R'_n}\max_{B'_1,..,B'_n} \sum_{i=1}^n \phi_i\left(\frac{B'_i}{R'_i}\right)
     = \sup_{B'_1,..,B'_n} \min_{R'_1,..,R'_n} \sum_{i=1}^n \phi_i\left(\frac{B'_i}{R'_i}\right)
\end{equation}
given
\begin{equation*}
    \sum_{i=1}^n B_i = B, \hspace{1cm} \sum_{i=1}^n R_i = R, \hspace{1cm} B_i, R_i \geq 0.
\end{equation*}
Then $B_i = c_i B$ and $R_i = c_i R$,
where 
\begin{equation*}
    c_i = \frac{\phi_i'(B/R)}{\sum_{j=1}^n \phi_j'(B/R)}.
\end{equation*}
\end{lemma}
\begin{proof}
Suppose that $B_1,..,B_n$, $R_1,..,R_n$, form a solution. Then in particular $B_1,..,B_n$ maximise
\begin{equation*}
    f: (x_1,..,x_n) \longmapsto \sum_{i=1}^n \phi_i\left(\frac{x_i}{R_i}\right)
\end{equation*}
given $0 = g(x_1,..,x_n) := x_1+..+x_n - B$. Using the theory of Lagrange multipliers, we see that then for all $i$:
\begin{equation}\label{eq:minimax1}
    \lambda = \left.\frac{\partial}{\partial x_i} f(x_1,..,x_n)\right|_{B_1,..,B_n} = \frac{1}{R_i} \phi_i' \left(\frac{B_i}{R_i}\right),
\end{equation}
where $\lambda\in \R$ is a fixed constant. Reasoning similarly starting from the RHS of~\eqref{eq:minimax}, we see that there is a Lagrange multiplier $\mu$ such that
\begin{equation}\label{eq:minimax2}
    \mu = \frac{-B_i}{R_i^2} \phi_i' \left(\frac{B_i}{R_i}\right).
\end{equation}
Plugging~\eqref{eq:minimax1} into \eqref{eq:minimax2}, we see that
\begin{equation*}
    \mu = \frac{-B_i}{R_i^2} \phi_i' \left(\frac{B_i}{R_i}\right) = \frac{-B_i}{R_i} \frac{1}{R_i} \phi_i' \left(\frac{B_i}{R_i}\right) = \frac{-B_i}{R_i} \lambda,
\end{equation*}
so
\begin{equation*}
    \frac{B_i}{R_i} = \frac{-\mu}{\lambda} := c
\end{equation*}
is constant. Using the boundary conditions, we obtain
\begin{equation*}
    B = \sum_{i=1}^n B_i = \sum_{i=1}^n c R_i = c \sum_{i=1}^n R_i = cR,
\end{equation*}
so $c = \tfrac{B}{R}$, so in particular $B_i = \tfrac{B}{R} R_i$. Plugging this back into~\eqref{eq:minimax1} and rearranging yields
\begin{equation*}
    R_i = \frac{1}{\lambda} \phi_i'\left(\frac{B}{R}\right).
\end{equation*}
Since $R_1,..,R_n$ must sum to $R$, this implies that 
\begin{equation*}
    \lambda = \frac{1}{R} \sum_{j=1}^n \phi_j'\left(\frac{B}{R}\right),
\end{equation*}
so $R_i = c_i R$ with $c_i$ defined as in the lemma statement. Using now that $B_i = \tfrac{B}{R} R_i$, this gives what we wanted.
\end{proof}

\subsection{Upper bound on the game value}\label{subsec:upper_bound}
To obtain an upper bound on the game value, we let Red use the strategy that is suggested by Lemma~\ref{lemma:minimax}.\\

\begin{lemma}[Upper bound]\label{lemma:upper}
    The game value is bounded from above by 
    \begin{equation*}
        \sum_{i=1}^n \phi_i\left(\frac{B}{R}\right).
    \end{equation*}
\end{lemma}
\begin{proof}
First we assume that $\tfrac{B}{R} \neq u_i$ for all $i$. This implies that each $\phi_i$ is differentiable at $B/R$, so we can set $R_i$ according to Lemma~\ref{lemma:minimax}:
\begin{equation*}
    R_i = R \frac{\phi_i'(B/R)}{\sum_{j=1}^n \phi_j'(B/R)}.
\end{equation*}
Now let $B_1,..,B_n$ be an arbitrary strategy of Lotto. Then, using a Taylor expansion for each $i$ of $\phi_i$ around $B/R$ and the concavity of the $\phi_i$'s, we obtain the following bound for the payoff:
\begin{eqnarray*}
    P(R_1,..,R_n,B_1,..,B_n) = \sum_{i=1}^n \phi_i\left(\frac{B_i}{R_i}\right) \leq \sum_{i=1}^n \left[\phi_i\left(\frac{B}{R}\right)
     + \phi_i'\left(\frac{B}{R}\right) \left(\frac{B_i}{R_i} - \frac{B}{R}\right) \right].
\end{eqnarray*}
Now substituting the expression for $R_i$ we see that
\begin{eqnarray*}
    \sum_{i=1}^n \phi_i'\left(\frac{B}{R}\right) \left(\frac{B_i}{R_i} - \frac{B}{R}\right) &=& 
    \sum_{i=1}^n \phi_i'\left(\frac{B}{R}\right) \frac{B_i \sum_j \phi_j'\left(\frac{B}{R}\right) }{R \phi_i'\left(\frac{B}{R}\right)} -
    \frac{B}{R} \sum_{i=1}^n \phi_i'\left(\frac{B}{R}\right) \\
    &=& \frac{1}{R} \sum_{i=1}^n B_i \sum_{j=1}^n \phi_j'\left(\frac{B}{R}\right)   - \frac{B}{R} \sum_{i=1}^n \phi_i'\left(\frac{B}{R}\right) = 0,
\end{eqnarray*}
where in the last equation we used that the $B_i$'s sum to $B$. This implies that 
\begin{equation*}
    H(B_1,..,B_n,R_1,..,R_n) \leq \sum_{i=1}^n \phi_i\left(\frac{B}{R}\right),
\end{equation*}
which is what we wanted.\\
\\
To deal with the case where $u_i = B/R$ for some (or possibly multiple $u_i$), change all such $u_i$ by at most $\epsilon$ to $u_i^*$. Now note that there is a constant $L$ such that each $\phi_i$ only changes by at most $L\epsilon$ (with respect to the uniform norm). This implies that the game value can change by at most $nL\epsilon$. Similarly the bound in this lemma only changes by at most a constant times $\epsilon$. Since this holds for all $\epsilon$, the bound of this lemma must also hold for $\epsilon = 0$.
\end{proof}

\noindent
Now we work out what this means in the case $\tfrac{B}{R} < 1$. We assume WLOG that $u_1\leq u_2 .. \leq u_n$. Suppose there is $0\leq k\leq n$ such that $u_k < \tfrac{B}{R} < u_{k+1}$ (where we set $u_0=0$ and $u_{n+1}=1$). In this case we obtain that
\begin{equation*}
\phi_i'\left(\frac{B}{R}\right) = 
    \begin{cases}
        \frac{w_i}{2} & i \leq k \\
        w_i & i > k
    \end{cases}
\end{equation*}
Denoting $w_{(k)} = \sum_{i\leq k} w_i$, this implies that
\begin{equation*}
    \sum_{i = 1}^n \phi_i'\left(\frac{B}{R}\right) = 1 - \frac{w_{(k)}}{2}.
\end{equation*}
The strategy for Red from Lemma~\ref{lemma:minimax} is therefore given by
\begin{equation*}
    R_i =
    \begin{cases}
        R \frac{w_i}{2 - w_{(k)}} & i \leq k \\
        R \frac{2w_i}{2 - w_{(k)}}& i > k 
    \end{cases}
\end{equation*}
The upper bound from Lemma~\ref{lemma:upper} now equals
\begin{eqnarray*}
    \sum_{i=1}^n \phi_i\left(\frac{B}{R}\right) &=& \sum_{i=1}^k \frac{w_i}{2} \left(u_i + \frac{B}{R}\right) + \sum_{i=k+1}^n w_i \frac{B}{R} \\
    &=& \frac{1}{2} \sum_{i=1}^k w_i u_i + \frac{B}{R} \left( 1 - \frac{w_{(k)}}{2}\right).
\end{eqnarray*}

\subsection{Lower bound on the game value}
As noted before, $GL\text{-}MS$ does not generally have a Nash equilibrium because of the non-convexity of $H$. This non-convexity is in turn caused by the non-convexity of $\phi_i(1/\cdot)$ for each $i$. To obtain a lower bound on the game value, we devise a different version of the game where this convexity does hold and find the optimal solution for that game.

\subsubsection{New payoff function}
To devise the new game, we take the following steps:
\begin{itemize}
    \item[i)] We take the functions $\phi_i$ and consider $\psi_i: x\mapsto \phi_i(1/x)$.
    \item[ii)] We then define $\psi_i^\dag$ as the lower convex envelope of $\psi_i$ (i.e. the largest convex function that is dominated by $\psi_i$). 
    \item[iii)] We transform $\psi_i^\dag$ back in the same way to define $\phi_i^\dag: x\mapsto \psi_i^\dag(1/x)$. This function is dominated by $\phi_i$ and is still concave. However, by construction $\phi_i^\dag(1/\cdot)$ is now convex.
\end{itemize}

\begin{figure}
    \centering
    \begin{subfigure}[b]{0.45\textwidth}
         \centering
         \includegraphics[width=\textwidth]{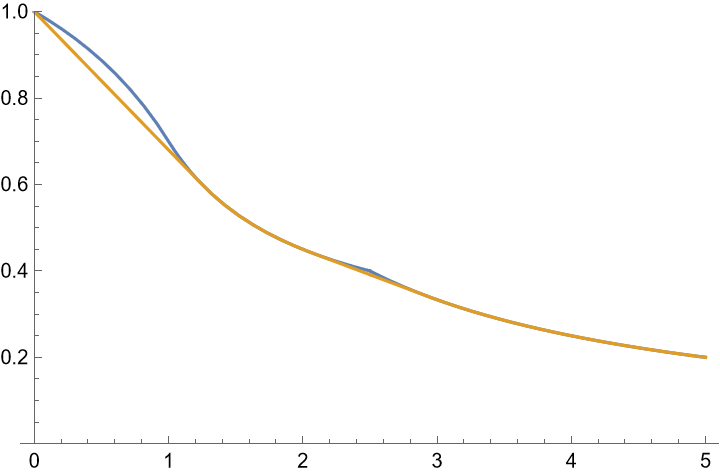}
         \caption{$\psi_i$ and $\psi_i^\dag$ for $u_i=\frac{2}{5}$}
     \end{subfigure}
     \hfill
     \begin{subfigure}[b]{0.45\textwidth}
         \centering
         \includegraphics[width=\textwidth]{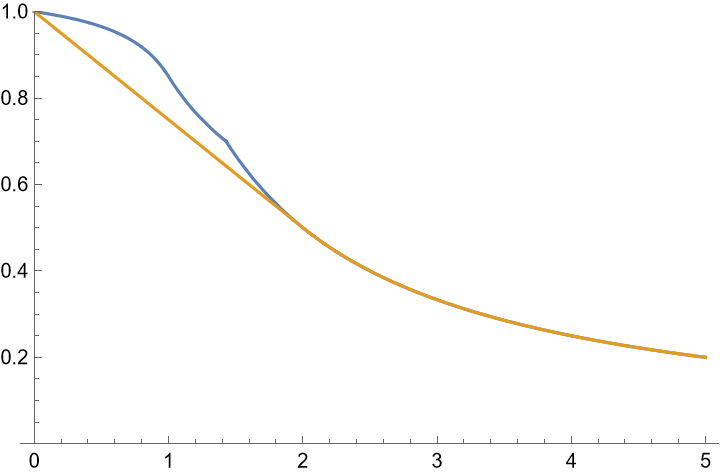}
         \caption{$\psi_i$ and $\psi_i^\dag$ for $u_i=\frac{7}{10}$}
     \end{subfigure}
    \caption{The function $\psi_i$ and its convex envelope $\psi_i^\dag$ for two different $u_i$, both with $w_i = 1$.}
    \label{fig:inv_func}
\end{figure}

We first follow this recipe for $u_i > 0$. In that case we get
\begin{equation*}
    \frac{1}{w_i} \psi_i(x) = 
    \begin{cases}
        1 - \frac{x(1-u_i)^2}{2(1-u_ix)} & x \leq 1 \\
        \frac{u_i}{2} - \frac{1}{2x} & 1 \leq x \leq \frac{1}{u_i} \\
        \frac{1}{x} & x \geq \frac{1}{u_i}
    \end{cases}.
\end{equation*}
To understand how to obtain the convex envelope $\psi_i^\dag$, the graph of $\psi_i$ is shown in Figure~\ref{fig:inv_func} for two different $u_i$, along with their convex envelopes. As can be observed from the graph, depending on $u_i$, there are two different versions of the convex envelope. The formulas are straightforward to derive. 
\begin{itemize}
    \item[i)] For $u_i\leq 2 - \sqrt{2}$ there is a linear part, then $u/2 + 1/(2x)$, another linear part and finally $1/x$. The equations are straightforward to derive and are given by
\begin{equation*}
    \frac{1}{w_i}  \psi_i^\dag(x) = 
    \begin{cases}
        1 - \frac{x}{2\alpha_i^2} & x \leq \alpha_i \\
        \frac{u}{2} + \frac{1}{2x} & \alpha_i \leq x \leq \beta_i \\
        \frac{2\gamma_i-x}{\gamma_i^2} & \beta_i \leq x \leq \gamma_i \\
        \frac{1}{x} & x \geq \gamma_i,
    \end{cases}
\end{equation*}
where
\begin{equation*}
    \alpha_i = \frac{2}{2-u_i}, \hspace{1cm} \beta_i = \frac{2(\sqrt{2}-1)}{u_i}, \hspace{1cm} \gamma_i = \frac{2(2-\sqrt{2})}{u_i}.
\end{equation*}
Note that for all $i$
\begin{equation*}
    1 \leq \alpha_i \leq \beta_i \leq \frac{1}{u_i} \leq \gamma_i. 
\end{equation*}
    \item[ii)] For $u_i \geq 2 - \sqrt{2}$ we get a linear part, followed by $1/x$:
\begin{equation*}
    \frac{1}{w_i}  \psi_i^\dag(x) = 
    \begin{cases}
        1 - \frac{x}{4}  & x \leq 2 \\
        \frac{1}{x} & x \geq 2
    \end{cases}.
\end{equation*}
\end{itemize}

\begin{figure}
    \centering
    \begin{subfigure}[b]{0.45\textwidth}
         \centering
         \includegraphics[width=\textwidth]{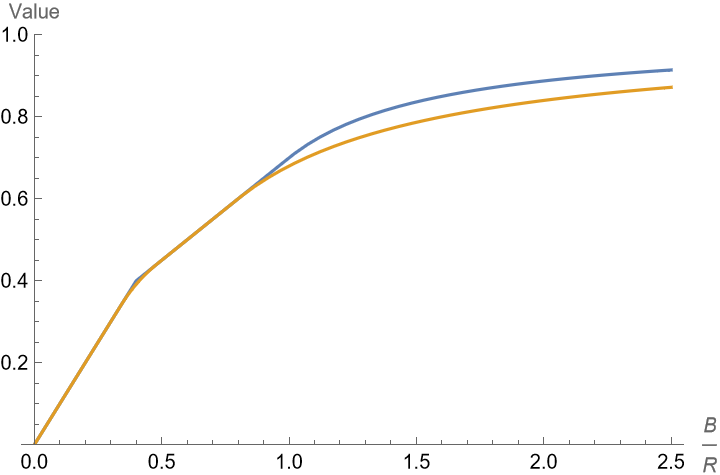}
         \caption{$\phi_i$ and $\phi_i^\dag$ for $u_i=\frac{2}{5}$}
     \end{subfigure}
     \hfill
     \begin{subfigure}[b]{0.45\textwidth}
         \centering
         \includegraphics[width=\textwidth]{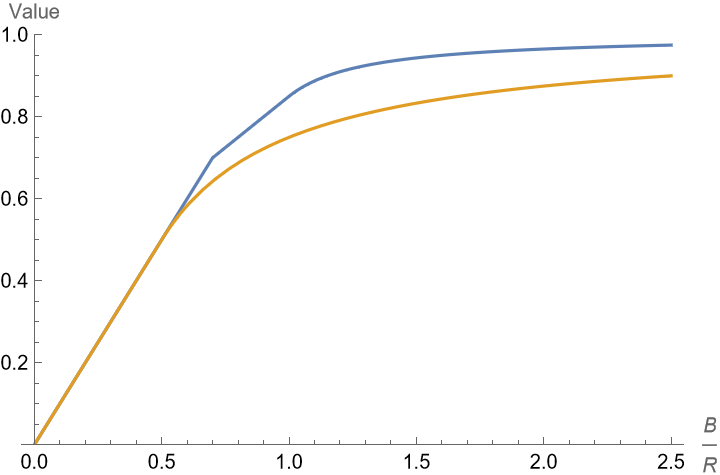}
         \caption{$\phi_i$ and $\phi_i^\dag$ for $u_i=\frac{7}{10}$}
     \end{subfigure}
    \caption{The function $\phi_i$ and the transformed convex envelope $\phi_i^\dag$ for two different $u_i$, both with $w_i = 1$.}
    \label{fig:newphi}
\end{figure}

\noindent
Now we transform this to $\phi_i^\dag$ by again considering $\psi_i^\dag(1/x)$. The two different versions are shown in Figure~\ref{fig:newphi}. The formulas are as follows. 
\begin{itemize}
    \item[i)] For $u_i\leq 2 - \sqrt{2}$
\begin{equation*}
    \frac{1}{w_i}  \phi_i^\dag(x) = 
    \begin{cases}
        x & x \leq \frac{1}{\gamma_i} \\
        \frac{2\gamma_ix-1}{\gamma_i^2x} & \frac{1}{\gamma_i} \leq x \leq \frac{1}{\beta_i} \\
        \frac{u}{2} + \frac{x}{2} & \frac{1}{\beta_i} \leq x \leq \frac{1}{\alpha_i} \\
        1 - \frac{1}{2\alpha_i^2x} & x \geq \frac{1}{\alpha_i} \\        
    \end{cases}
\end{equation*}
    \item[ii)] For $u_i \geq 2 - \sqrt{2}$ 
\begin{equation*}
    \frac{1}{w_i}  \phi_i^\dag(x) = 
    \begin{cases}
        x & x \leq \frac{1}{2}  \\
        1 - \frac{1}{4x}  & x \geq \frac{1}{2}
    \end{cases}.
\end{equation*}
\end{itemize}

This was all in case $u_i>0$. For $u_i = 0$, we see analogously that
\begin{equation*}
    \frac{1}{w_i} \phi_i(x) =
    \begin{cases}
        \frac{x}{2} & x \leq 1 \\
        1 - \frac{1}{2x} & x \geq 1
    \end{cases}
\end{equation*}
so
\begin{equation*}
    \frac{1}{w_i} \psi_i(x) =
    \begin{cases}
        1 - \frac{x}{2} & x \leq 1 \\
        \frac{1}{2x} & x \geq 1.
    \end{cases}
\end{equation*}
Note that this function is already convex. Therefore $\psi_i^\dag = \psi_i$ and $\phi_i^\dag = \phi_i$.\\
\\
By construction the $\phi_i^\dag$ are now reciprocally concave, i.e. $\phi_i^\dag$ is concave and $\phi_i^\dag(1/\cdot)$ is convex. We use these to define a new payoff function 
\begin{equation*}
    H^\dag(B_1,..,B_n,R_1,..,R_n) = \sum_{i=1}^n \phi_i^\dag\left(\frac{B_i}{R_i}\right).
\end{equation*}
Note that $H^\dag$ is concave in $B_1,..,B_n$ and convex in $R_1,..,R_n$. Also note that setting $\phi_i(B_i/0) = 1$ and $B_i\geq\epsilon>0$ ensures continuity of $H^\dag$ everywhere.

\subsubsection{Value of the dominated game}\label{subsubsec:domgame}
We now define a new game with payoff $H^\dag$, which has the required convex-concavity. First fix $B,R > 0$ and set
\begin{equation*}
    c_i = \frac{(\phi_i^\dag)'(B/R)}{\sum_{j=1}^n (\phi_j^\dag)'(B/R)}.
\end{equation*}
We fix $\epsilon$ so that $0 <\epsilon < B \min_i c_i$ and require that $B_i \geq \epsilon$ for all $i$. Note that this extra restriction can only lower the value for Blue, so the value of the resulting game is still a lower bound for the original game.\\
\\
All the $\phi_i^\dag$'s are $C^1$, so we can apply Lemma~\ref{lemma:minimax} to conclude that the unique Nash equilibrium is obtained by setting $B_i = c_i B$ and $R_i = c_i R $. Moreover, the game value is $\sum_i \phi_i^\dag(B/R)$. Note that by construction of $\epsilon$, the $B_i$ all satisfy the restriction $B_i \geq \epsilon$.\\ 
\\
To compute the $c_i$'s more explicitly, one needs the derivative of $\phi_i^\dag$, which is given below.
\begin{itemize}
    \item[i)] For $u_i = 0$
\begin{equation*}
    \frac{1}{w_i} (\phi_i^\dag)'(x) =
    \begin{cases}
        \frac{1}{2} & x \leq 1 \\
        \frac{1}{2x^2} & x \geq 1
    \end{cases}
\end{equation*}
    \item[ii)] For $0\leq u_i\leq 2 - \sqrt{2}$
\begin{equation*}
    \frac{1}{w_i}  (\phi_i^\dag)'(x) = 
    \begin{cases}
        1 & x \leq \frac{1}{\gamma_i} \\
        \frac{1}{\gamma_i^2x^2} & \frac{1}{\gamma_i} \leq x \leq \frac{1}{\beta_i} \\
        \frac{1}{2} & \frac{1}{\beta_i} \leq x \leq \frac{1}{\alpha_i} \\
        \frac{1}{2\alpha_i^2x^2} & x \geq \frac{1}{\alpha_i}     
    \end{cases}
\end{equation*}
    \item[iii)] For $u_i \geq 2 - \sqrt{2}$ 
\begin{equation*}
    \frac{1}{w_i}  (\phi_i^\dag)' = 
    \begin{cases}
        1 & x \leq \frac{1}{2}  \\
        \frac{1}{4x^2}  & x \geq \frac{1}{2}
    \end{cases}
\end{equation*}
\end{itemize}

\subsection{Value of the original game}
So far we have obtained an upper bound and a lower bound for the game. Now note that in some cases, the $c_i$'s as computed in Section~\ref{subsubsec:domgame} coincide with the ones that were used in Section~\ref{subsec:upper_bound}. Namely when for each $i$ one of the following conditions holds:
\begin{itemize}
    \item[i)] $u_i = 0$
    \item[ii)] $0\leq u_i\leq 2 - \sqrt{2}$ and either $$\frac{B}{R} \leq 1/\gamma_i = \frac{u_i}{2(2-\sqrt{2})}$$ or $$\frac{u_i}{2(\sqrt{2}-1)} = 1/\beta_i \leq \frac{B}{R} \leq 1/\alpha_i = 1 - u_i/2$$
    \item[iii)] $u_i \geq 2 - \sqrt{2}$ and $\tfrac{B}{R} \leq 1/2$
\end{itemize} 
Therefore, in these cases the upper and lower bound coincide. This implies that the value of the game equals 
\begin{equation*}
    \sum_{i=1}^n \phi_i\left(\frac{B}{R}\right),
\end{equation*}
since this is the upper bound we found. 
In particular, when all $u_i$'s are $0$, the value is 
\begin{equation*}
    \frac{1}{2} \frac{B}{R}.
\end{equation*}
We get back to this in Section~\ref{sec:interpretation}.
In the other cases the bounds do not coincide, so we only obtain a range of possible game values. Figure~\ref{fig:upper_lower} shows the lower and upper bounds for a fixed choice of $w_i, u_i$ as a function of $B/R$.\\

\begin{figure}[!ht]
    \centering
    \begin{subfigure}[b]{0.45\textwidth}
         \centering
         \includegraphics[width=\textwidth]{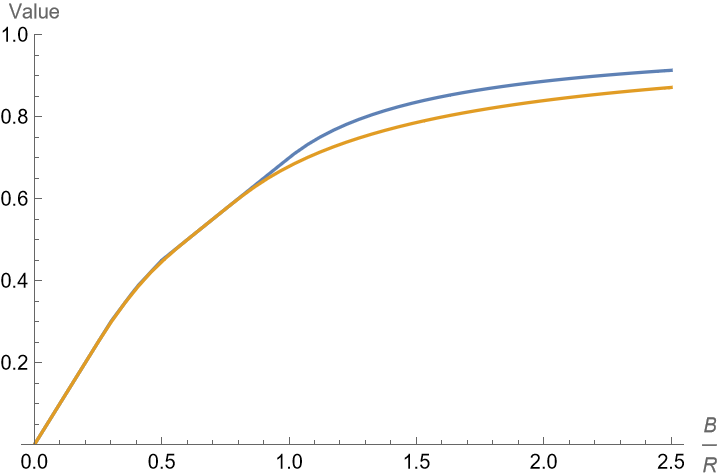}
         \caption{$u_1 = \frac{3}{10}, u_2 = \frac{4}{10}, u_3 = \frac{5}{10}$}
     \end{subfigure}
     \hfill
     \begin{subfigure}[b]{0.45\textwidth}
         \centering
         \includegraphics[width=\textwidth]{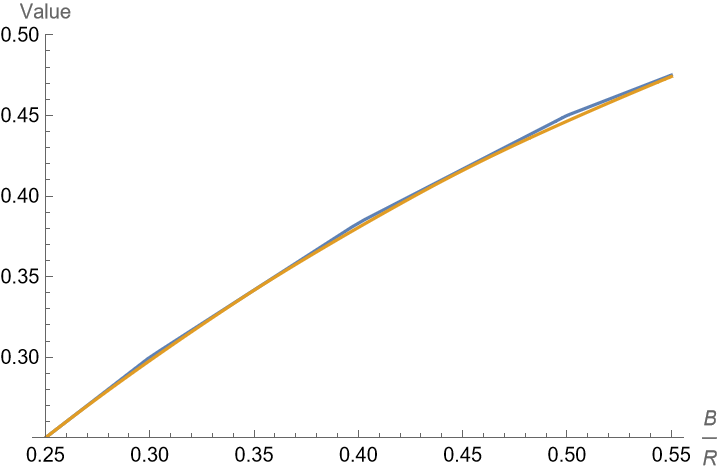}
         \caption{$u_1 = \frac{3}{10}, u_2 = \frac{4}{10}, u_3 = \frac{5}{10}$}
     \end{subfigure}\\
     \begin{subfigure}[b]{0.45\textwidth}
         \centering
         \includegraphics[width=\textwidth]{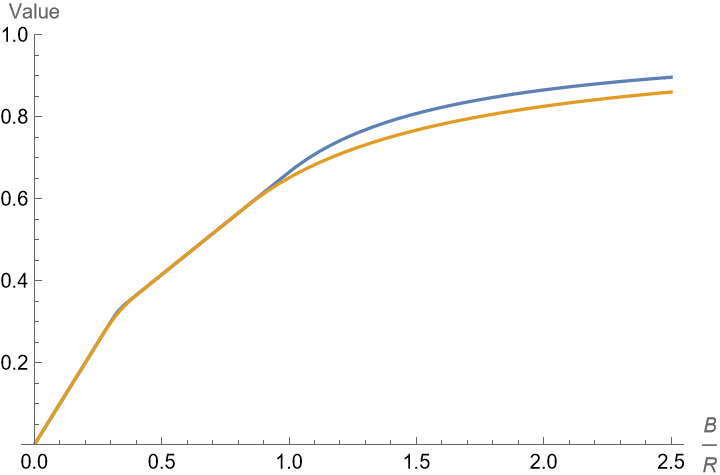}
         \caption{$u_1=0.31,u_2=0.33,u_3=0.35$}
     \end{subfigure}
     \hfill
     \begin{subfigure}[b]{0.45\textwidth}
         \centering
         \includegraphics[width=\textwidth]{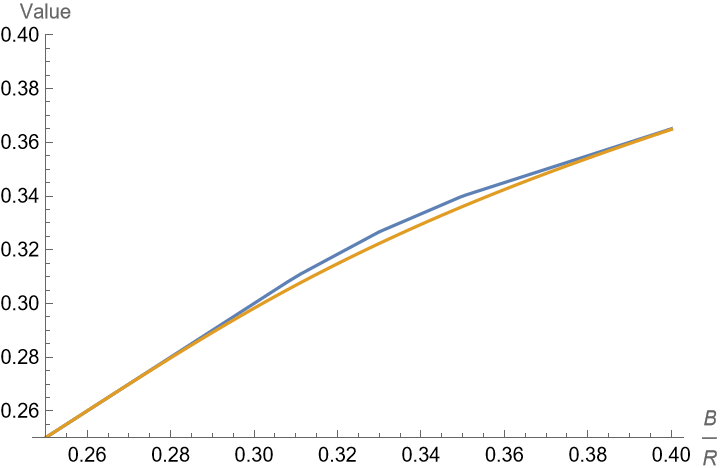}
         \caption{$u_1=0.31,u_2=0.33,u_3=0.35$}
     \end{subfigure}
    \caption{The value of the GL-MS with 3 fields, for different values of $u_i$. All three fields have value 1/3. The right figures are zoomed-in versions of the left figures.}
    \label{fig:upper_lower}
\end{figure}

\begin{remark}[Comparison with known results]
When all $u_i$'s are set to $0$, we obtain the following unique second-stage strategy: $B_i = w_i B$ and $R_i = w_i R$. Then on each field a single field General Lotto game is played with the same resource ratio $B_i/R_i = B/R$. The total expected payoff still equals $\tfrac{B}{2R}$ whenever $B\leq R$ and $1 - \frac{R}{2B}$ otherwise. This is (like in Remark~\ref{rmk:sf}) consistent with Corollary 1 of \citet{kovenock2021generalizations}. The difference between our case and theirs is that they treat the game as a one-shot game. Apparently changing the game to a multistage variant did not affect the value.
\end{remark}
\section{Measuring information versus strength} \label{sec:Measuring}
In the preceding sections, we have determined the optimal strategies and value for games with fixed amounts of resources and information, effectively addressing our first main research question. For this section we transition to our second main research question: The influence of information versus strength and the weapons mix problem. The weapons mix problem consists of choosing the optimal combination of information (scouting drones) and resources which provide strength (troops).
In this section we present three different ways to compare the influence of information against the influence of resources for $GL\text{-}S$, ending with the weapons mix problem. Two of these ways are similar to methods that were independently used in \cite{chandan2022strategic}.

\subsection{Influence ratio}
One way to measure the influence of information versus strength is to compare the increase in game value under the addition of information against the increase of the game value under the addition of resources. To this end we fix $R>0$ and compute for each pair of $B, u$ the ratio of derivatives of the game value with respect to $B$ and $u$, respectively. This yields the following:
\begin{eqnarray*}
    V_u &:=& \frac{\partial V}{\partial u} = 
    \begin{cases}
        0 & \frac{B}{R} < u \\
        \frac{1}{2} & u < \frac{B}{R} < 1 \\
        \frac{(1-u)(2B/R - 1- u)}{2\left(\tfrac{B}{R}-u\right)^2} & 1< \frac{B}{R}
    \end{cases}\\
    V_B &:=& \frac{\partial V}{\partial B} = 
    \begin{cases}
        \frac{1}{R} & \frac{B}{R} < u \\
        \frac{1}{2R} & u < \frac{B}{R} < 1 \\
        \frac{(1-u)^2}{2R(B/R-u)^2} & 1< \frac{B}{R}
    \end{cases}
\end{eqnarray*}
This allows us to define the influence ratio IR:
\begin{eqnarray*}
    IR &:=& \frac{1}{R} \frac{V_u}{V_B} = 
    \begin{cases}
        0 & \frac{B}{R} < u \\
        1 & u < \frac{B}{R} < 1 \\
        1 + \frac{2}{1-u} \left(\tfrac{B}{R}-1\right) & 1< \frac{B}{R}
    \end{cases}
\end{eqnarray*}
The graph of the influence ratio is shown for $R=1$ in Figure~\ref{fig:logratio}. The diagonal line $u=B$ between blue and purple marks the `Exploitation Line', i.e. under this line additional information is valuable because it can be used, but above this line there are simply not enough resources. Using this, the graph can be divided into 3 parts. Directly above the exploitation line we have the upper triangular region, where the influence ratio is 0, which means adding information does not increase the value of the game. Directly below the exploitation line is the lower triangular region, where the influence ratio is exactly 1, which means everywhere in the lower triangle, increasing either $u$ or $B$ by the same amount, leads to the same increase in value. And on the right side we have the right rectangle, where the influence ratio is larger than 1, which means the value increases more for an increase in $u$ than for the same increase in $B$. Note the interesting behaviour around the point $(1,1)$. The reason is that for $B<1$, going from information $u=1-\epsilon$ to $u=1$ is not valuable at all, because there are not enough resources to use the information. However, for $B>1$ the increase of information from $1-\epsilon$ to $1$ is extremely valuable. The reason for this is that when $u=1-\epsilon$ adding resources will never yield a value of $1$, where adding just the final $\epsilon$ bit of information does yield game value $1$.

\begin{figure}[ht]
    \centering
    \includegraphics[width=1.2\textwidth]{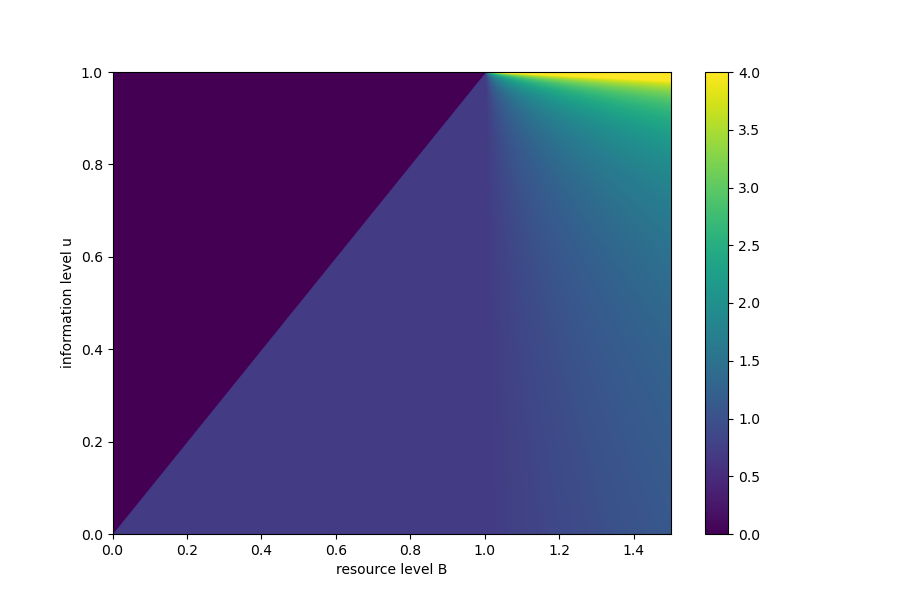}
    \caption{Heatmap of the Influence Ratio, as a function of resources $B$ and information $u$. A higher influence ratio indicates that information is more valuable.}
    \label{fig:logratio}
\end{figure}

\subsection{Resources needed to attain a given game value}
Another way to quantify the influence of information is by studying how many resources are needed to obtain a certain game value given a level of information. The resource ratio $B/R$ that is needed to attain game value $v$ given information $u$ equals
\begin{equation*}
    \frac{B}{R} = 
    \begin{cases}
        u + \frac{(1-u)^2}{2(1-v)} & u \leq 2v - 1 \\
        2v - u & 2v-1 \leq u \leq v \\
        v & u \geq v
    \end{cases}
\end{equation*}
Note that for any given fixed value $v$, the equation above exactly represents the contour line of the value function for this particular value $v$. In Figure~\ref{fig:contour}, we display several of these contour lines.
This visualization also provides insight into how to maximize the game value with minimal additional information or resource allocation. Specifically, this can be achieved by moving perpendicular to the contour lines.
Notably, in the top-left region, we observe once again that adding information does not enhance the game value; instead, only the addition of resources can lead to an increase.

\begin{figure}[ht]
    \centering
    \includegraphics[width=\textwidth]{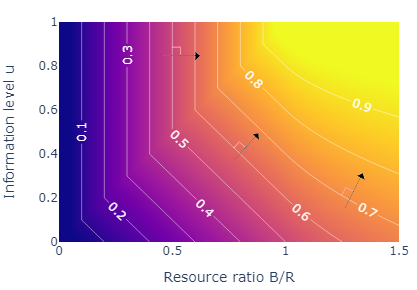}
    \caption{Contour plot of the value of $GL\text{-}S$ as a function of $B/R$ and $u$. The arrows indicate the direction in which the value increases most steeply.}
    \label{fig:contour}
\end{figure}

\subsection{Optimising information vs strength}\label{subsec:budget}
In this section, we delve into the weapons mix problem, which revolves around a fixed budget $D$ to be divided between resources and information. We assume $R=1$ and set the cost of resources $B$ to $1$ per unit and the cost of information at $c$ per unit. This yields the restriction that $B + cu \leq D$. Using the solution found to $GL\text{-}S$ in Section \ref{sec:singlefield}, we can determine the $B$ and $u$ that maximise the game value under this constraint, thereby identifying the optimal allocation of the budget between resources and information.\\
\\
Varying the total budget $D$ produces the graphs in Figure~\ref{fig:budget}. The line corresponding to $c=100$ was included to illustrate the case in which information is so expensive that it is not used at all, resembling the `no-information case'.\\
\\
The top-left graph of Figure~\ref{fig:budget} shows the maximum attainable value for a budget $D$. It's evident that a higher cost $c$ of information always leads to a lower game value.\\
\\
Moving to the top-right graph, we see that for $c<1$ it is optimal for any budget $D$ to invest in both resources and information, with both steadily increasing for higher budgets until a game value of $1$ is reached. For $c>1$ it is better to only invest in resources for when the budget is small. However, if we increase the budget, there is a certain point after which it also becomes worthwhile to invest in information as well. Interestingly, from that point on Blue will actually start decreasing the amount spent on resources, as his investment in information increases faster than the increase in budget. The decrease in resources bought can be seen in the bottom-left graph.\\
\\
Finally, note from the top-left graph that a budget $D\geq 1+c$ ensures the the game value is $1$. This outcome is understandable since $c$ can be spent on buying information, thereby acquiring perfect information $u=1$. With perfect information, Blue only needs to spend $R=1$ on resources to always match Red and thus always win. Since a budget of $1+c$ guarantees Blue always wins, the surplus budget need not be used. The unused budget is shown in the bottom-right graph.

\begin{figure}[!ht]
    \centering
    \includegraphics[width=\textwidth]{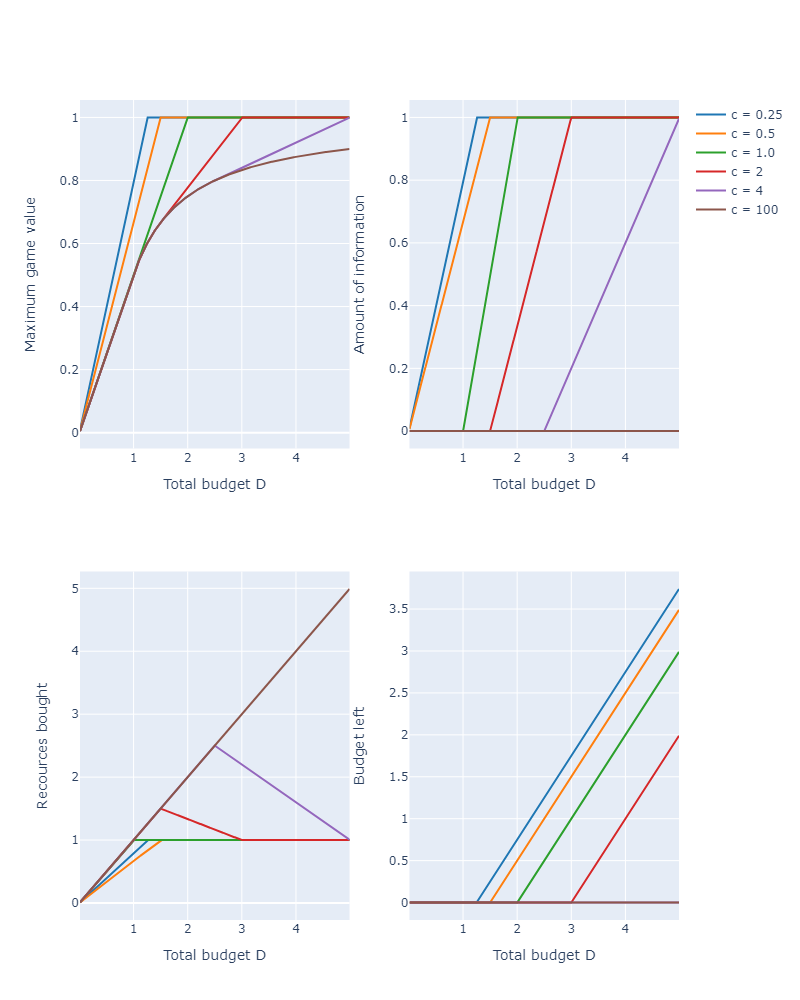}
    \caption{Let $D$ be the fixed budget to be spent between information and resources. Here we plot the characteristics as function of the total budget $D$, for several different costs of information $c$.}
    \label{fig:budget}
\end{figure}
\section{Interpretation}\label{sec:interpretation}
In the previous sections we have tackled the two main research questions outlined at the start of this paper: the weapons mix problem and efficient strategies when using scouts. Now, our focus shifts towards interpreting these questions and their respective solutions within a military context, particularly concerning the scenario of troops providing firepower and scouting drones providing information. In this section we outline general guidelines for optimal strategies, the most effective combination of troops and scouting drones, and whether to prioritise additional troops or additional scouting drones. Underlying these guidelines are the insights obtained in various parts of the paper. We use the strategies and values calculated for the $GL\text{-}S$ model and employ the three distinct ways to compare information and strength within these models.\\
\\

\subsection{Efficient strategies with scouts}
We start with guidelines for efficient strategies when the number of resources and the amount of information are predetermined. The situation we analyse is as follows: Red and Blue will face each other in a number of battles. Both have a given amount of resources and additionally Blue has a given number of scouts, which roughly translates to a probability to gain information about Red's allocation of resources.\\ 
\\
Blue and Red have different challenges in formulating their strategies.
For Blue, the strategy can be split into two sub-strategies: what to do if Blue gains information, and what to do if Blue doesn't gain information. How does Blue best convert his scouting advantage into the battle?
Red aims to provide a strategy that is as scout-resistant as possible, recognising the possibility of her allocation to be scouted to Blue. Is this achieved by playing more stable and always sending the same number of resources, which causes Blue to learn nothing extra when he gains information? Or is this done by playing more risky and pooling a large part of resources together for a large attack, hoping that this one is not revealed? \\
\\
The answer to all of these questions depends on the probability of Blue gaining information and the resource ratio $B/R$. There are three distinct cases, each with its own set of strategies for Blue and Red. We discuss these cases separately and provide insights based on Theorem \ref{thm:GL-S}, which describes an optimal solution to $GL\text{-}S$.
%
%

\subsubsection*{Red heavily outnumbers Blue, who has many scouts \texorpdfstring{$(B/R<u)$}{B/R<u}}
In this case Blue's optimal strategy involves only allocating resources when information about Red's allocation is obtained, conserving resources otherwise. Blue does not have enough resources to match Red every time the field is revealed, but should try to match as often as possible. 
Red's best response is to always allocate the exact same number of resources, thus neutralising the advantage Blue has when he gets information about Red's allocation.

\subsubsection*{Red outnumbers Blue, who has few scouts \texorpdfstring{$(u<B/R<1)$}{u<B/R<1}}
Here Red uses the same strategy as she would in a standard General Lotto Game, always allocating resources but randomising the exact amount slightly. Blue can match Red’s allocation whenever it is revealed and still have resources left in his resource budget. Therefore, Blue should also allocate resources occasionally when he has no information about Red's allocation. If Blue decides to allocate resources, he should follow Red’s strategy for determining the exact number. 

\subsubsection*{Blue outnumbers Red \texorpdfstring{$(1<B/R)$}{1<B/R}}
Now Blue can match Red’s allocation whenever it is revealed and additionally allocate resources every time it is not revealed. The exact amount should be randomised slightly. For Red, it is best to allocate resources even less often than in a standard General Lotto game. This means that Red should often not allocate any resources, and occasionally allocate a large number, while still randomising the exact amount. The level of risk taken by Red increases with both the ratio with which she is outnumbered and how often Blue gains information on her resource allocation. If Red decides to allocate resources, she should follow Blue’s strategy for determining the exact number.

\subsection{Weapons mix problem}
Now that we have provided guidelines for efficient strategies given fixed amounts of resources and information, we discuss guidelines for when the amounts of resources and information are not fixed, and can instead be chosen under a fixed budget. First we first discuss the influence of information and resources, respectively, and how they relate to military practice and experience. Then we consider the budget constraints, thereby addressing the weapons mix problem.

\subsubsection*{The value of information} 
To analyse the value of information, recall Figure~\ref{fig:Lotto}. Two modes of behaviour exist, depending on whether Blue has more $(B/R > 1)$ or less $(B/R<1)$ resources than Red. In the first case, an increase in information always leads to to an increase of value. 
However, in the second case, there is a cutoff point: initially, extra information is beneficial, but beyond the cutoff point, there is no additional gain. 
The main takeaway from these considerations is that information only provides value as long as there are enough resources in strength available to take advantage of this information.\\
\\
We see something similar in military practice, for instance in the war following the Russian invasion of Ukraine. Ukraine has a great intelligence position due to help from many allies. However, at some point more intelligence does not yield any better results. Knowing exactly what the opponent does is not helpful when there are no resources to counter them.

\subsubsection*{The value of resources}
For a deeper analysis of the value of resources, we turn to Figure~\ref{fig:single_field_value}. Again, there is a clear break point. When $B/R$ is smaller than $u$, all of the Blue forces are used to counter Red forces when they are observed, hence they are always employed effectively. However, when $B/R$ is larger than $u$, some of the Blue resources are used in situations where Red is not observed. These Blue resources are used less effectively. This results in a smaller increase in game value after the break point.\\
\\
Similarly, in military practice, deploying troops is much more effective when the strength of the opponent in the area is known. As soon as all available information on the opponent has been exploited, any new troops are used in more uncertain situations and are therefore less effective. We find this reflected in the Russian-Ukrainian war, where Russia has many more troops and weaponry to deploy but has an inferior intelligence position. This makes the deployment of Russian troops less effective.

\subsubsection*{Budget constraints}
Finally we arrive at our second research question: how to optimally distribute a military budget over resources and information. We studied this question most directly in Section~\ref{subsec:budget}, of which the main results are shown in Figure~\ref{fig:budget}. From this figure, we can extract three regimes, which we show in Table~\ref{tab:scenarios}. 
\begin{itemize}
    \item[(I)] In this case information is cheap compared to resources. This corresponds to $c<1$ in Figure~\ref{fig:budget}. In this case it is directly worthwhile to buy information, so as budget increases resources and information are bought equally.
    \item[(II)] Here information is expensive (so $c>1$) and there is a low budget. In this case it is not worthwhile to invest in information and it is best to spend all of the budget on resources.
    \item[(III)] As in case (II) information is expensive (so $c>1$), but there is also a high budget. Now it is again worthwhile to buy information. In fact, with a higher budget some resources are exchanged for more information (as one can tell from the decreasing lines in the bottom-left graph of Figure~\ref{fig:budget}).
\end{itemize}
The main takeaway here is that if information is cheap, you should always buy it, along with resources. If information is expensive, you should focus on resources first and only buy information if your budget is high enough.

\begin{table}[ht]
\centering
\begin{tabular}{ll|ll}
\multicolumn{2}{l|}{\multirow{2}{*}{}} & \multicolumn{2}{l}{\hspace{1.2cm}Information}                                           \\
\multicolumn{2}{l|}{}                         & \multicolumn{1}{l|}{\textbf{Cheap}}                        & \textbf{Expensive}                  \\ \hline
\multirow{2}{*}{Budget}         & \textbf{Low}         & \multicolumn{1}{l|}{\multirow{2}{*}{(I) Combine}} & (II) Only resources           \\ 
                                & \textbf{High}        & \multicolumn{1}{l|}{}                             & (III) Focus on information
\end{tabular}
\caption{Preferred strategy in cases of high and low budget and with cheap and expensive information. Derived from Figure~\ref{fig:budget}.}
\label{tab:scenarios}
\end{table}

\section{Discussion} \label{sec:conclusion}
In this paper, we introduced the General Lotto Game with Scouts, a resource allocation game enhanced with the possibility of gaining information on the opposing player though scouting. We provided a complete solution for the game concerning a single field. Subsequently, we extended our analysis to a multistage version, establishing upper and lower bounds for its value and determining conditions on when they coincide. Additionally, we introduced several metrics to quantify the trade-off between information and resources. Finally, we interpreted these findings and drew qualitative conclusions, summarised as follows:
\begin{itemize}
    \item Information is only useful when there are enough resources to exploit it.
    \item Resources are employed more efficiently in the presence of information.
    \item When information is cheap, one should always buy it alongside resources. When it is expensive, one should first focus on buying resources and only invest in information if the budget is high enough.
    \item When Blue receives information, it is always best to call as often as possible.
    \item Red should play very stable if she heavily outnumbers Blue, and very risky if Blue heavily outnumbers her.
\end{itemize}
In this section we discuss the model and results. We conclude with some suggestions for future work.

\subsection{Discussion of the model and results}

\subsubsection*{Lotto constraint before or after observing Red's resource allocation}
The characterising feature of a General Lotto model is the constraint that the total allocated budget should equal a fixed constant in expectation.
Since our studied situation involves potential information revelation, one faces a decision when this constraint should be applied: one can require the expected total resources to equal a constant before or after potentially revealing Red's resource allocation.\\
\\
The second option might seem more natural, but has certain curiosities in the General Lotto framework. For instance: whenever Red allocates less resources than Blue's total budget, she always loses. Therefore Red will likely alternate between allocating nothing and allocating significantly more than Blue's budget. Additionally, Blue has only one sensible response when Red's allocation is revealed, namely matching Red with the highest possible probability. \\
\\
In this paper, we chose for the former definition: requiring the budget constraint before the allocation is revealed. This allows Blue to decide how much of the budget he wants to allocate in the information and no-information case, which in turn allows us to study where he can use his budget most effectively.
The only drawback is that since Blue's allocation depends on Red's, his allocation distributions depends on Red's allocation distribution.


\subsubsection*{Constant price of information}
In Section~\ref{subsec:budget}, we examined a fixed budget $D$ to be divided between resources and information. We assigned information a fixed cost of $c$ per unit. However, this straightforward approach implies that the cost of increasing information from 0 to 0.1 is identical to the cost of increasing from 0.9 to 1. In practical scenarios, obtaining basic information about a situation is often much cheaper than acquiring the final details.\\
\\
This choice for a fixed cost model significantly impacts the perceived value of information in comparison to resources. In particular in the top-right part of Figure~\ref{fig:logratio}, the value of information is much larger than the value of resources. 
This large difference stems from the fact that even a marginal increase in information can elevate the game value to 1, a feat unachievable solely through an increase in resources.\\
\\
To better reflect real-world conditions, one could consider assigning information a non-linear cost.
For instance, introducing a cost function such as $F(u) = 1/(1-u)-1$ would result in the cost of acquiring information increasing with each unit purchased. 
Consequently, perfect information becomes impossibly expensive to attain, thereby imposing a practical limit on information acquisition.

\subsection{Future work}
We conclude this work with potential avenues for extending our research beyond its current scope. While we have already mentioned adjusting the Lotto constraint and the introduction of a non-linear price of information, there are several additional options for future research.\\
\\
First, one could consider more complicated scenarios for the single field case. For instance, one could investigate scenarios where the information obtained by Blue is partial, perhaps by revealing each of Red’s units independently with certain probabilities. Another intriguing possibility is to introduce false signals, complicating the accuracy of the information gained by Blue and adding an element of uncertainty.\\
\\
Second, it could be interesting to study extensions to the multistage version of the game. For example, one could explore allowing Blue to optimise the detection probabilities $u_1,..,u_n$ under a fixed budget constraint. Additionally, each field could have a `visibility' parameter, which influences how hard it is for Blue to obtain information, resulting in varying information costs per field. Another option is to give Blue a fixed amount of scouts, which he must then distribute strategically among the fields. I.e. if blue has only two scouts, only the two fields he sends these to could be revealed. These extensions might give a `search game dynamic' where a trade-off will be sought between allocating resources and search capacity to more and less visible fields.

\bibliography{refs.bib}

\end{document}